\documentclass[12pt,english]{article}
\usepackage[round]{natbib}
\usepackage{babel}
\usepackage[margin=1.1in]{geometry}
\usepackage{amsmath}
\usepackage{amssymb}
\usepackage{fullpage}
\usepackage{bbm}
\usepackage{enumerate}
\usepackage{setspace}
\usepackage{sgame}
\usepackage[colorlinks = true,linkcolor = blue,urlcolor  = blue,citecolor = blue,anchorcolor = blue]{hyperref}
\usepackage{amsmath,amssymb}
\usepackage{amsthm}
\usepackage{amsfonts}
\usepackage{tikz}
\usepackage{float}
\usepackage{amsmath,mathtools}
\usepackage{fmtcount}
\usepackage[affil-it]{authblk}

\begin{document}

\title{Coordination through ambiguous language}
\author{Michele Crescenzi\footnote{I wish to thank Hannu Vartiainen for invaluable support and supervision, and Hannu Salonen and Mark Voorneveld for helpful comments. Financial support from the OP Group Research Foundation is gratefully acknowledged. All errors are mine.}\\michele.crescenzi@helsinki.fi}
\affil{University of Helsinki and Helsinki GSE}

\date{April 2022}

\newtheorem{proposition}{Proposition}
\newtheorem{assumption}{Assumption}
\newtheorem{definition}{Definition}
\newtheorem{lemma}{Lemma}
\newtheorem{remark}{Remark}
\newtheorem{corollary}{Corollary}

\maketitle
\onehalfspacing

\begin{abstract}
We provide a syntactic construction of correlated equilibrium. For any finite game, we study how players coordinate their play on a signal by means of a public strategy whose instructions are expressed in some natural language. Language can be ambiguous in that different players may assign different truth values to the very same formula in the same state of the world. We model ambiguity using the player-dependent logic of \cite{halpern2015}. We show that, absent any ambiguity, self-enforcing coordination always induces a correlated equilibrium of the underlying game. When language ambiguity is allowed, self-enforcing coordination strategies induce subjective correlated equilibria.

\bigskip
{\sl JEL CLASSIFICATION}: C72, D82, D83
\bigskip

\end{abstract}
\newpage

\section{Introduction}
Correlated equilibrium is a solution concept that captures the impact of communication on strategic interaction. It does so without modeling explicitly the communication process in which players are involved. Differently put, correlated equilibrium ``express[es] an assumption that players have implicit communication opportunities, in addition to the strategic options explicitly described in the game model'' \cite[p. 245]{myerson1991book}. All \textit{implicit} communication opportunities are subsumed into canonical correlating devices which send private recommendations on how to play the game. But such devices, and the corresponding ``equilibria[,] may have no natural interpretation'' \cite[p. 47]{osborne-rubinstein1994}.

Our goal in this paper is to provide an alternative construction of correlated equilibrium which, we believe, has a more natural interpretation than the canonical one. The main idea behind our construction is that correlated play can be induced by the ambiguity of the natural language through which players communicate. By ambiguity, we mean interpretive uncertainty stemming from the fact that words or sentences can have a plurality of meanings. Let us consider an example. Suppose a central banker delivers the following public speech: ``If the GDP growth is sustained, then interest rates will be kept constant; otherwise they will be lowered''. Firms listen to the speech in order to decide on their investments, which depend on future interest rates. But what is the true content of the banker's statement? More specifically, how should the antecedent ``if the GDP growth is sustained'' be interpreted? Is there a threshold $x$ such that if the actual growth rate $y$ is greater than $x$ then it is really the case that ``the GDP growth is sustained''? One can argue that such a threshold $x$ does exist but, unless its value is explicitly stipulated in some contract or convention, it is not necessarily unique. A firm $i$ may think that the threshold is $x_i$, whereas firm $j$ believes that it is $x_j \neq x_i$. In addition, either firm does not know what threshold the other is using to classify the growth rate as sustained or not. In sum, even if the banker's speech is public, it may convey differential information to those who hear it. Differently put, uncertainty about the interpretation of an ambiguous statement acts like a correlating device that sends private messages to players, so inducing differential information.

The backbone of our analysis is the following process of communication and coordination. In a given simultaneous-move game, players receive information about the prevailing state of the world at the pre-play stage. Information is not payoff relevant, and it may come in different varieties. It can be a public speech, a private signal, a sunspot, etc. Players have the opportunity to condition their play on the information received in the pre-play stage. They do so by means of a coordination strategy, which is a public list of conditional statements on how to play the game. In the banker's example, a coordination strategy can contain the following statements: ``if interest rates will be kept constant, then only firm $i$ invests'' and ``if interest rates will be lowered, everybody invests''. Firms publicly agree to follow the action recommendations contained in the coordination strategy. But due to language ambiguity, firms do not know with certainty how others will interpret the strategy recommendations and, therefore, how they will react to the banker's speech. As in the standard case, this uncertainty can sustain equilibrium payoffs that are outside the convex hull of Nash equilibrium payoffs of the underlying game.

Contrary to the standard construction of correlated equilibrium, we separate messages from their meaning. We do so by modeling explicitly the language through which players communicate. A language is a set of well-defined formulas that describe every relevant aspect of the world. To capture ambiguity, we use the logic of \cite{halpern2015}. In it, truthfulness of formulas is defined relative to a player. Consequently, there can be states of the world where different players give different truth values to the very same formula. Players can disagree on a subset of formulas, namely those constructed as the conjunction or negation of primitive propositions, whereas the interpretation of probability formulas, i.e. beliefs, is the same for everybody. This means that every player is sophisticated enough to understand that others might be using different information partitions to form their beliefs.

Our main contribution is to provide a syntactic construction of correlated equilibrium. We consider two cases. In the first, we model the communication and coordination process illustrated in the central banker's example under the assumption that language is not ambiguous. We show that, for any finite game, any self-enforcing coordination strategy induces an objective correlated equilibrium distribution of the underlying game. In addition, any objective correlated equilibrium distribution of the underlying game can be induced by some coordination strategy in some unambiguous epistemic structure capturing players' interpretations of formulas in the language. In the second case, we allow language to be ambiguous. We obtain the same characterization as in the unambiguous case with the proviso that equilibrium distributions are now subjective correlated equilibria. We thus show that language ambiguity provides a justification for heterogeneous beliefs about strategic play.

We illustrate the model in Section \ref{sec:Model}. It consists of three parts: the syntax (how formulas are formed), the semantics (how meaning to formulas is assigned), and the coordination process. Results are presented in Section \ref{sec:Results}, where the two cases of common-interpretation and ambiguous epistemic structures are treated separately.

\subsection{Related literature}
Correlated equilibrium is introduced in \cite{aumann1974}. A reformulation of it in a decision-theoretic framework is provided in \cite{aumann1987}. Our analysis is related to the following strands of the literature.

First, a classical literature initiated by \cite{forges1988sunspots, forges1990universal} and \cite{barany1992} studies whether and how correlated equilibrium can be obtained in a decentralized manner, i.e. without the help of a mediator. In our analysis, a mediator is not strictly necessary in that the information that players receive in the pre-play stage can be interpreted as a sunspot \textit{\`{a} la} \cite{cass1983sunspots}. 

\cite{lehrer1996}, \cite{lehrer-sorin1997}, and \cite{di-tillio2004} study public mediated talk in which correlation is achieved trough a machine that receives private inputs and sends out public recommendations. If we assume that the information in our model is provided by a mediator, then communication is always one-way, i.e. from the mediator to the players. Under this interpretation, correlation is achieved through uncertainty about the messages sent by the mediator. Players do not need to exchange messages with each other, nor do they need to send reports to the mediator.

\cite{blume-board2013} examine strategic interaction under the assumption that players differ in their ``language competence'', i.e. their ability to use language. They model language explicitly. Other analyses aimed at modeling ambiguity (or vagueness) in natural language include \cite{lipman2009} and \cite{jaegher2003}. However, none of these papers use the syntactic approach as we do.

Our work is also related to the literature on epistemic foundations of solution concepts initiated by \cite{aumann-brandenburger1995}. The main goal of this literature is to find epistemic conditions that give rise to standard solution concepts. The approach is to model explicitly how players reason about the game and, in particular, how they reason about the rationality of their opponents. Recent contributions in which correlated equilibrium is studied include \cite{bach-perea-CE} and \cite{barelli2009}. Rather than rational play, the focus of our analysis is on how players reason about the realization of extraneous signals, and how this reasoning is affected by language ambiguity. All the contributions mentioned so far are carried out from a set-theoretic perspective. But another branch of the research program on epistemic foundations uses techniques from modal logic, as is done in \cite{lorini2010modal} and \cite{galeazzi-lorini2016}. An extensive overview is provided in \cite{de-bruin2010book}. To the best of our knowledge, our analysis would be the first to use modal logic to examine ambiguity about the interpretation of extraneous signals in games. As we already mentioned, we build on the logic of \cite{halpern2015}. In particular, the syntax (Section \ref{syntax}) and the semantics (Section \ref{semantics}) are theirs.
\section{Model}\label{sec:Model}
Let $G=\left(I, (A_i, u_i)_{i\in I}\right)$ be a finite game with simultaneous moves. The set of players is $I=\{1, \dots, n\}$. For every $i\in I$, $A_i$ is a non-empty, finite set of actions, and $u_i: \times_{j\in I} A_j \longrightarrow \mathbb{R}$ is the corresponding payoff function. As is standard, we define $A:= \times_{i\in I} A_i$ and, for any $i$, $A_{-i}:= \times_{j\neq i} A_j$.

Players coordinate their play in $G$ on the realizations of a payoff-irrelevant signal. In the pre-play stage, they agree on a list of instructions that tell them how to play the game conditional on signal observations. Players' reasoning about the game and the signals is captured by a formal language, which we are going to model explicitly. We describe the syntax in subsection \ref{syntax}, the semantics in subsection \ref{semantics}, and the coordination strategy in subsection \ref{coordination}.

\subsection{Syntax}\label{syntax}
The fundamental object is a non-empty, countable set $\Phi$ of primitive propositions, with typical elements $p,q, \dots$. Propositions in $\Phi$ describe non-epistemic aspects of the world.  A \textit{language} $\mathcal{L}(\Phi)$ is a set of well-formed formulas constructed from $\Phi$ through syntactic rules. Since no confusion should arise, from now on we omit the reference to $\Phi$ and write $\mathcal{L}$. The formulas contained in $\mathcal{L}$ determine the expressiveness of the language, i.e. the set of epistemic and non-epistemic aspects of the world that players can reason about. We construct $\mathcal{L}$ according to the following syntax:
\begin{itemize}
\item If $p \in \Phi$, then $p$ is a formula in $\mathcal{L}$;
\item \textit{Negation}: If $\varphi \in \mathcal{L}$, then $\lnot \varphi$ (``not $\varphi$'') is a formula in $\mathcal{L}$;
\item \textit{Conjunction}: If $\varphi, \psi \in \mathcal{L}$, then $\varphi \land \psi$ (``$\varphi$ and $\psi$'') is a formula in $\mathcal{L}$;
\item \textit{Probability formulas}: If $\varphi_1, \dots, \varphi_k \in \mathcal{L}$ and $b_1, \dots, b_k , c \in \mathbb{R}$, then, for every $i\in I$,
\begin{equation*}
b_1 \mathsf{pr}_i(\varphi_1) + \dots + b_k \mathsf{pr}_i(\varphi_k)  \geq c
\end{equation*}
is a formula in $\mathcal{L}$. The intended reading of $\mathsf{pr}_i(\varphi) \geq x$ is ``the probability that player $i$ ascribes to formula $\varphi$ is at least $x$'';
\item \textit{Modal operator} $\mathsf{CB}$: If $\varphi \in \mathcal{L}$, then $\mathsf{CB} \varphi$ (``it is commonly believed that $\varphi$'') is a formula in $\mathcal{L}$.
\end{itemize}

Probability formulas allow players to reason about beliefs and expected payoffs. We also want $\mathcal{L}$ to be sufficiently rich to describe how agents play the game $G$ and how they interpret the signals that they observe. Hence we assume that, for every $i\in I$, and for every $a_i \in A_i$, there is a primitive proposition $\mathsf{pl}_i a_i$ in $\Phi$. The intended reading of $\mathsf{pl}_i a_i$ is ``$i$ chooses $a_i$'' or, equivalently, ``$i$ plays $a_i$''. We assume that all these propositions describing choices are distinct elements, i.e. if $\mathsf{pl}_i a_i = \mathsf{pl}_j b_j$, then $i=j$ and $a_i=b_j$. Let $\Phi_G$ be the finite subset of $\Phi$ containing all such propositions about choices in $G$. In order to describe signals, let $\Phi^*$ be the set obtained by closing off $\Phi \backslash \Phi_G$ under negation and conjunction. Notice that formulas in $\Phi^*$ describe non-epistemic aspects of the world that are not payoff-relevant. We assume that there is a finite subset $\Sigma \subseteq \Phi^*$ of signals. Furthermore, if $\sigma \in \Sigma$, then $\left\lbrace \mathsf{rec}_i \,\sigma: i\in I\right\rbrace \subseteq \Phi$. The intended reading of $\mathsf{rec}_i \, \sigma$ is ``$i$ has received signal $\sigma$''.

We make use of the following \textbf{abbreviations}:
\begin{itemize}
\item \textit{Implication}: $\varphi \implies \psi$ (``$\varphi$ implies $\psi$'') is an abbreviation for $\lnot \left(\varphi \land \lnot \psi \right)$;
\item \textit{Belief operator}: $\mathsf{B}_i \varphi$ (``$i$ believes that $\varphi$'') is an abbreviation for
\begin{equation*}
\left(\mathsf{pr}_i(\varphi) \geq 1 \right) \land \left(-\mathsf{pr}_i (\varphi) \geq -1\right);
\end{equation*}
\item \textit{Mutual belief operator}: $\mathsf{EB}\varphi$ (``everybody believes that $\varphi$'') is an abbreviation for $\land_{i\in I} \mathsf{B}_i \varphi$. In addition, we define $\mathsf{EB}^m \varphi$ (``$\varphi$ is $m$th-order mutual belief'') recursively: $\mathsf{EB}^1 \varphi = \mathsf{EB} \varphi$, and $\mathsf{EB}^m \varphi = \mathsf{EB}\left(\mathsf{EB}^{m-1}\varphi\right)$ for $m \geq 2$;
\item $U_i(a_i)$ is the abbreviation for the probability formula
\begin{equation}\label{eq:Exp-Payoff}
\sum_{(a_1, \dots, a_{i-1}, a_{i+1}, \dots, a_n)\in A_{-i}} u_i(a_1, \dots, a_n)\, \mathsf{pr}_i\left(\mathsf{pl}_1 a_1 \land \cdots \land \mathsf{pl}_{i-1} a_{i-1} \land  \mathsf{pl}_{i+1}a_{i+1} \land \cdots \land \mathsf{pl}_n a_{n} \right).
\end{equation}
The intended reading of $U_i(a_i)$ is ``the expected payoff to $i$ from playing $a_i$''. In order for this intended reading to be meaningful, the probabilities that $i$ ascribes to formulas in \eqref{eq:Exp-Payoff} must be non-negative and sum up to one. Under the assumptions we make in Subsections \ref{semantics} and \ref{coordination}, it is always the case that these probabilities are well-defined, so making the reading of $U_i(a_i)$ as expected payoff unproblematic.
\item $a_i^* = a_i$ is the abbreviation for
\begin{equation*}\label{eq:U-max}
\land_{a_i'\in A_i} \left(U_i(a_i)\geq U_i(a_i')\right).
\end{equation*}
The intended reading of $a_i^* = a_i$ is ``$a_i$ is utility-maximizing''.
\item $\mathsf{rat}_i$ is the abbreviation for 
\begin{equation}\label{eq:Rat-i}
\land_{a_i \in A_i} \left(\mathsf{pl}_i a_i \implies (a_i^* = a_i) \right).
\end{equation}
The intended reading of $\mathsf{rat}_i$ is ``$i$ is rational''. Notice that \eqref{eq:Rat-i} is equivalent to saying that $i$ never chooses an action that is not utility-maximizing.
\end{itemize}

\subsection{Semantics}\label{semantics}
We need a semantic model to assign meaning to formulas in $\mathcal{L}$. That is, we need a consistent set of rules to determine whether any given formula is true or false. The semantic model we use is an \textit{epistemic probability structure} in which the interpretation of primitive propositions is player-dependent. Formally, an epistemic probability structure $M$ over $\Phi$ is a tuple $\left(\Omega, \mu, \{\pi_i\}_{i\in I}, \{H_i\}_{i\in I}\right)$, where:
\begin{itemize}
\item $\Omega$ is a non-empty, finite set of states or possible worlds;
\item $\mu$ is a common prior on (the power set of) $\Omega$;
\item $\pi_i: \Omega \times \Phi \longrightarrow \{0,1\}$ is agent $i$'s interpretation function. Agent $i$ deems proposition $p$ as true in state $\omega$ if $\pi_i(\omega, p) = 1$, and false otherwise; 
\item $H_i$ is agent $i$'s information partition over $\Omega$, with typical element $h_i$. We write $h_i(\omega)$ to indicate the cell containing the states that $i$ considers as possible when the true state is $\omega$. We make assumptions on how information partitions are determined by signals at the end of this subsection.
\end{itemize}

The finiteness of $\Omega$ is without loss of generality since we are confining ourselves to finite games. The player-dependent interpretation function $\pi_i$ captures language ambiguity: in a given world, different agents may assign different truth values to the very same primitive proposition. If $\pi_i=\pi_j$ for all $i,j\in I$, then we say that $M$ is a \textit{common-interpretation} structure. The latter corresponds to the standard case without ambiguity where the interpretation of every formula is player-independent. If $M$ is not a common-interpretation structure, then we call it \textit{ambiguous}.

Agents update beliefs through Bayes's rule. Given any event $E\subseteq \Omega$, agent $i$'s posterior belief about $E$ at $\omega$ is
\begin{equation*}
\mu\left(E \lvert h_i(\omega)\right) =
\frac{\mu(E \cap h_i(\omega))}{\mu(h_i(\omega))}.
\end{equation*}
To ensure that posteriors are always well-defined, we assume that $\mu(h_i(\omega)) >0$ for every state $\omega \in \Omega$ and every player $i \in I$.

Meaning to formulas in a structure $M$ is given inductively. The expression $(M, \omega, i) \vDash \varphi$ means that $\varphi$ holds at $\omega$ according to player $i$ in structure $M$. In addition, the intension of a formula $\varphi$ to player $i$ is $[[\varphi]]_i := \left\lbrace \omega \in \Omega: (M, \omega, i) \vDash \varphi\right\rbrace$, i.e. the set of states where $i$ deems $\varphi$ as true in structure $M$. Meaning to formulas is given as follows:
\begin{itemize}
\item If $p$ is a primitive proposition in $\Phi$, then $(M, \omega,i) \vDash p$ iff $\pi_i(\omega, p) = 1$;

\item $(M, \omega,i) \vDash \varphi \land \psi$ iff $(M, \omega,i) \vDash \varphi$ and $(M, \omega,i) \vDash \psi$;

\item $(M, \omega,i) \vDash \lnot \varphi$ iff $(M, \omega,i) \not\vDash \varphi$;

\item $(M, \omega,i) \vDash b_1 \mathsf{pr}_j(\varphi_1) + \dots + b_k \mathsf{pr}_j(\varphi_k)  \geq c$ iff 
\begin{equation}\label{eq:ProbFormulas}
b_1\, \mu\left([[\varphi_1]]_j \lvert h_j(\omega)\right) + \dots + b_k\, \mu\left([[\varphi_k]]_j \lvert h_j(\omega)\right) \geq c;
\end{equation}

\item $(M, \omega,i) \vDash \mathsf{B}_j \varphi$ iff $\mu\left([[\varphi]]_j \lvert h_j(\omega)\right) = 1$;

\item $(M, \omega,i) \vDash \mathsf{CB} \varphi$ iff $(M, \omega,i) \vDash \mathsf{EB}^k \varphi$ for $k=1,2,\dots$.
\end{itemize}

We emphasize that meaning to a formula is always given relative to a player. Due to language ambiguity, there can be states where different players assign different meaning to the very same formula. Formally, there can be states and formulas such that $(M,\omega, i)\vDash \varphi$ and $(M,\omega, j)\vDash \lnot\varphi$ for some $i$ and $j$. However, players are fully sophisticated in that they understand that others are using different information partitions to update beliefs\footnote{This is the \textit{innermost-scope semantics} of \cite{halpern2014}.}. Consequently, everybody agrees on the interpretation of probability formulas. As can be seen from \eqref{eq:ProbFormulas}, according to player $i$, agent $j$ assigns probability at least $c$ to a formula $\varphi$ if and only if the set of worlds where $\varphi$ holds according to $j$ has probability at least $c$ according to $j$. When the interpretation of a formula $\varphi$ is player-independent at a state $\omega$, we simplify notation and write $(M,\omega) \vDash \varphi$ instead of $(M,\omega,i) \vDash \varphi$ for all $i \in I$. In addition, we write $M \vDash \varphi$ when $(M,\omega,i)\vDash \varphi$ for every $\omega \in \Omega$ and every $i\in I$. In this case, we also say that $\varphi$ is \textit{valid} in $M$.

We now make two assumptions about the interpretation of signals and information partitions.
\begin{assumption}\label{ass:Sign-interpret}
For every $i,j \in I$,
\begin{itemize}
\item the collection
\begin{equation*}
\left\lbrace [[\mathsf{rec}_i \, \sigma]]_j: \sigma \in \Sigma \text{ and }[[\mathsf{rec}_i \, \sigma]]_j \neq \emptyset \right\rbrace
\end{equation*}
is a partition of $\Omega$;
\item for every $\sigma, \sigma' \in \Sigma$, if $[[\mathsf{rec}_i \, \sigma]]_j = [[\mathsf{rec}_i \, \sigma']]_j \neq \emptyset$, then $\sigma = \sigma'$.
\end{itemize}
\end{assumption}

The assumption says that, according to any player, everyone receives one, and only one, signal at every state. Because of ambiguity, the event of $i$'s receiving signal $\sigma$ can be interpreted differently by different agents. For instance, it could be the case that $(M,\omega,i) \vDash \mathsf{rec}_i \, \sigma \land \lnot\mathsf{rec}_i \, \sigma'$ and $(M,\omega,j) \vDash \mathsf{rec}_i \, \sigma' \land \lnot\mathsf{rec}_i \, \sigma$, where $\sigma \neq \sigma'$. For ease of reference, we write $\sigma_{i,\omega}$ to denote the necessarily unique signal that $i$ thinks she is observing at state $\omega$.

\begin{assumption}\label{ass:info-partitions}
For every $i\in I$ and every $\omega \in \Omega$,
\begin{equation*}
h_i(\omega)=[[\mathsf{rec}_i \, \sigma_{i,\omega}]]_i.
\end{equation*}
\end{assumption}

The assumption says that a player's information is determined by the signal she thinks she is observing. More specifically, the worlds that $i$ considers as possible at $\omega$ are all those where $i$ thinks that she is observing the same signal as in $\omega$. Notice that $\omega \in h_i(\omega)$ and, if $\omega' \in h_i(\omega)$, then $\sigma_{i,\omega} = \sigma_{i,\omega'}$.

The following example is meant to illustrate how one can use the main concepts introduced so far to capture ambiguity.
\paragraph{Example 1} There are two agents: $A$(nn) and $B$(ob). 
Suppose there is a primitive proposition $p \in \Phi$, whose intended reading is ``the air temperature is extreme''. According to Ann, temperatures are extreme if they are at most $x_A$ or at least $y_A$. According to Bob, temperatures are extreme if they are at most $x_B$ or at least $y_B$. Suppose $x_A < x_B < y_A < y_B$. The set of possible states of the world is represented in Table \ref{table:ex1_mod}.

\begin{table}[ht]
\centering
\begin{tabular}{c|ccc}
State & Temperature & Ann & Bob\\
\hline
$\omega_1$  & $x_A$ & $\mathsf{rec}_A\, p$  & $\mathsf{rec}_B \, p$\\
$\omega_2$  & $y_A$ & $\mathsf{rec}_A \,p$ & $\mathsf{rec}_B\, \lnot p$\\
$\omega_3$  & $x_B$ & $\mathsf{rec}_A \,\lnot p$ & $\mathsf{rec}_B\, p$\\
$\omega_4$  & $\frac{x_B + y_A}{2}$& $\mathsf{rec}_A \, \lnot p$ & $\mathsf{rec}_B \,\lnot p$\\
\end{tabular}
\caption{States of the world}
\label{table:ex1_mod}
\end{table}

Each state is a complete description of all the epistemic and non-epistemic aspects of the world. In state $\omega_1$, the actual temperature is $x_A$. Therefore, the proposition $p$ is deemed as true by both Ann and Bob. But in state $\omega_2$, $p$ is true according to Ann and false according to Bob. Their disagreement stems from language ambiguity. Since $p$ can be given a plurality of meanings, different agents may interpret it differently. We emphasize that, in ambiguous structures, there can be primitive propositions whose interpretation is not ambiguous at all. For instance, suppose that also the primitive proposition $q$ is in $\Phi$, where $q$ stands for ``the air temperature is $x_A$''. This proposition is unambiguous, and both Ann and Bob interpret it as true in state $\omega_1$ and false otherwise.

The true state of the world is observed through signals. Suppose that the set of signals is $\Sigma = \{p, \lnot p\}$. In addition, each player receives $\sigma \in \Sigma$ in a given state if and only if he or she deems $\sigma$ as true in that state. Each row of Table \ref{table:ex1_mod} indicates the signals received by either player in the corresponding state. We assume that the interpretation of formulas of the form $\mathsf{rec}_i \sigma$ is not ambiguous. For instance, we have $(M,\omega_2) \vDash \mathsf{rec}_A \,p \land \mathsf{rec}_B \lnot p$ even if $(M,\omega_2, A) \vDash p$ and $(M, \omega_2, B)\vDash \lnot p$. In words, Ann thinks at $\omega_2$ that Bob receives the signal ``the air temperature is not extreme'' while she thinks that the temperature is actually extreme. We use the formulas of the form $\mathsf{rec}_i \sigma$ to obtain the following information partitions:
\begin{align*}
H_A &= \left\lbrace \{\omega_1, \omega_2\}, \{\omega_3, \omega_4\} \right\rbrace\\
H_B &= \left\lbrace \{\omega_1, \omega_3\}, \{\omega_2, \omega_4\} \right\rbrace.
\end{align*}

Suppose that the common prior $\mu$ is uniform over $\Omega$. We now want to make a few remarks on how agents form beliefs. We start by noticing that 
\begin{align*}
(M, \omega_1, A) & \vDash \mathsf{B}_A \, p \; \land \; \mathsf{B}_B \, p\\
(M, \omega_1, B) & \vDash \mathsf{B}_A \, p \; \land \; \mathsf{B}_B \, p,
\end{align*}
or, in compact notation, $(M, \omega_1) \vDash \mathsf{EB} \, p$. That is, everybody believes that $p$ at $\omega_1$. This follows from the fact that
\begin{align*}
\mu \left([[p]]_A\vert h_A(\omega_1)\right) &= \mu \left(\{\omega_1, \omega_2\}\vert \{\omega_1, \omega_2\}\right) = 1\\
\mu \left([[p]]_B\vert h_B(\omega_1)\right) &= \mu \left(\{\omega_1, \omega_3\}\vert \{\omega_1, \omega_3\}\right) = 1.
\end{align*}
However, it holds that $(M, \omega_1)\vDash \lnot \mathsf{B}_A \mathsf{B}_B \, p$. In words, Ann does not believe that Bob believes that $p$. Indeed we have
\begin{equation*}
\mu \left([[\mathsf{B}_B \, p]]_A\vert h_A(\omega_1)\right) = \mu \left(\{\omega_1, \omega_3\}\vert \{\omega_1, \omega_2\}\right) = \frac{1}{2}.
\end{equation*}
Therefore, even if Ann and Bob receive the very same signal ``the air temperature is extreme'' in state $\omega_1$, it is not common belief between them that this is indeed the case. More specifically, the formula $p$ (and the formulas $\mathsf{rec}_A\, p$ and $\mathsf{rec}_B\, p$) is a first-order mutual belief at $\omega_1$, but it is not a second-order mutual belief. \textit{A fortiori}, $p$ is not commonly believed. This shows how ambiguity generates higher-order uncertainty in the interpretation of formulas. Things would be different if the epistemic structure had common interpretation. Suppose that both agents have the same interpretation function as in Ann's column in Table \ref{table:ex1_mod}. It is then immediate that, in state $\omega_1$, the formula $p$ (and the formulas $\mathsf{rec}_A\, p$ and $\mathsf{rec}_B\, p$) is not just first-order mutual belief but also common belief. Formally, $(M, \omega_1) \vDash \mathsf{CB}\, p$.

\subsection{Coordination}\label{coordination}
We now describe how agents make choices in $G$. We start by assuming the following.

\begin{assumption}\label{ass:measure}
In any structure $M$, for every $i\in I$ and every $a_i \in A_i$,
\begin{equation*}
M \vDash \mathsf{pl}_i a_i \implies \land_{a_i'\neq a_i}\left(\lnot \mathsf{pl}_i a_i'\right).
\end{equation*}
\end{assumption}

The assumption simply says that (it is commonly believed that) everyone does not play more than one action in each state.

Agents have the opportunity to coordinate their choices in $G$ through signals in $\Sigma$. More specifically, they can devise a \textit{coordination strategy} $C$ that tells them how to play $G$ depending on the realizations of signals in $\Sigma$.

\begin{definition}[Coordination strategy]\label{def:Coord}
A coordination strategy $C$ is a finite subset of $\mathcal{L}$ such that:
\begin{enumerate}
\item for each player $i\in I$ and each signal $\sigma \in \Sigma$, there is a unique action $a_i\in A_i$ such that the formula $\mathsf{rec}_i \, \sigma \implies \mathsf{pl}_i a_i$ belongs to $C$;
\item for every $\varphi \in C$, $M \vDash \varphi$.
\end{enumerate}
\end{definition}

A coordination strategy is a finite list of conditional propositions of the following form: ``if $i$ receives signal $\sigma$, then $i$ plays action $a_i$'', ``if $j$ receives signal $\sigma'$, then $j$ plays action $a_j$'', and so on. Notice that a strategy associates every signal with one, and only one, action for each player, but different signals may be associated with the same action recommendation. The strategy is public in that every formula contained in it is valid in $M$, hence it is common knowledge among everyone.

A coordination strategy is a set of instructions. Definition \ref{def:Coord} ensures that such a set is complete, i.e. it provides everyone with an action recommendation for every signal realization that can possibly be observed. But it says nothing about the rationality, or lack thereof, of these recommended actions. Therefore, we want to restrict our analysis to epistemic structures, and coordination strategies, that meet minimal rationality requirements.

\begin{assumption}[Individual rationality]\label{ass:IR}
In any structure $M$, for every $i\in I$, it holds that $(M,i)\vDash \mathsf{rat}_i$.
\end{assumption}

The assumption says that every $i$ chooses an action only if she deems it utility-maximizing. In other words, it is always true, according to player $i$, that $i$'s choices are utility-maximizing. As a consequence, $i$ always believes in her own rationality, and it is commonly believed that it is so. When the underlying epistemic structure has common interpretation, Assumption \ref{ass:IR} is tantamount to assuming common belief in rationality, i.e. common belief in the event that everyone is rational.

\begin{remark}\label{remarkCB}
Let $M$ be a structure satisfying Assumption \ref{ass:IR}. Then we have:
\begin{enumerate}
\item $M \vDash \mathsf{CB} \left(\land_{i\in I} \mathsf{B}_i \left(\mathsf{rat}_i\right)\right)$;
\item If $M$ is a common-interpretation structure, then $M \vDash \mathsf{CB} \left(\land_{i\in I} \mathsf{rat}_i\right)$.
\end{enumerate}
\end{remark}
\begin{proof}
By Assumption \ref{ass:IR} and the definition of the belief operator, we have that, for every $i\in I$, $(M,i)\vDash \mathsf{B}_i \left(\mathsf{rat}_i\right)$. Since the interpretation of probability formulas is player-independent, the latter is equivalent to $M\vDash \mathsf{B}_i \left(\mathsf{rat}_i\right)$ for every $i\in I$. Therefore, the formula $\land_{i\in I} \mathsf{B}_i \left(\mathsf{rat}_i\right)$ is valid in $M$, and it is always common belief that it is a true formula.

Now suppose that $M$ is a common-interpretation structure. Thus we have that, for every $i,j\in I$, $(M,i)\vDash \mathsf{rat}_i$ if and only if $(M,j)\vDash \mathsf{rat}_i$. But then it is immediate to get $M \vDash \land_{i\in I} \mathsf{rat}_i$, from which the result follows.
\end{proof}

In an individually rational structure, any coordination strategy $C$ is self-enforcing in that no one has the incentive to disobey its action recommendations. Formally, for every $\omega$ and every $i$, there exists a unique signal $\sigma$ such that
\begin{equation*}
(M, \omega, i) \vDash \mathsf{rec}_i \sigma \land \mathsf{pl}_i a_i \land (a_i^* = a_i),
\end{equation*}
where $a_i$ is the action prescribed by the formula $\mathsf{rec}_i \sigma \implies \mathsf{pl}_i a_i$ in $C$. To see why this is the case, observe the following. First, Assumptions \ref{ass:Sign-interpret} and \ref{ass:info-partitions} assure that there is a unique signal $\sigma = \sigma_{i,\omega}$ such that $(M, \omega, i) \vDash \mathsf{rec}_i \, \sigma$. Second, this signal is associated to a unique action by the coordination strategy $C$: there is a unique action $a_i$ with $\mathsf{rec}_i \, \sigma \implies \mathsf{pl}_i a_i$ in $C$ such that $M \vDash \mathsf{rec}_i \, \sigma \implies \mathsf{pl}_i a_i$. Third, the previous two points together imply $(M,\omega, i)\vDash \mathsf{pl}_i a_i$. Finally, by Assumption \ref{ass:IR}, we also get $(M,\omega, i)\vDash (a_i^* = a_i)$.

\section{Results}\label{sec:Results}
Our goal is to characterize the probability distributions over $A$ that are induced by a given coordination strategy. Formally, every coordination strategy $C$ induces a profile $\left(\gamma_i\right)_{i\in I}$ of probability distributions over $A$. For every $i$, we define
\begin{equation}\label{eq:OutDistr}
\gamma_i(a_1, \dots, a_n) := \mu \left([[\mathsf{pl}_1 a_1 \land \cdots \land  \mathsf{pl}_n a_n]]_i\right) = \mu \left(\{\omega: (M,\omega,i)\vDash \mathsf{pl}_1 a_1 \land \cdots \land  \mathsf{pl}_n a_n\}\right).
\end{equation}

Each $\gamma_i$ is a well-defined probability distribution. First, it is clear from \eqref{eq:OutDistr} that $\gamma_i (a)\geq 0$ for every $a\in A$. Second, by Assumptions \ref{ass:Sign-interpret} and \ref{ass:info-partitions}, and Definition \ref{def:Coord}, for every $\omega$ there exists an action profile $a\in A$ such that $(M,\omega,i)\vDash \mathsf{pl}_1 a_1 \land \cdots \land  \mathsf{pl}_n a_n$. Third, by Assumption \ref{ass:measure}, $a\neq a'$ implies that 
\begin{equation*}
[[\mathsf{pl}_1 a_1 \land \cdots \land  \mathsf{pl}_n a_n]]_i \cap [[\mathsf{pl}_1 a_1' \land \cdots \land  \mathsf{pl}_n a_n']]_i = \emptyset.
\end{equation*}
Therefore we have that $\sum_{a\in A} \gamma_i(a)=1$. From now on, when no confusion should arise, we abuse notation and write $\mathsf{pl} a$ instead of $\mathsf{pl}_1 a_1 \land \cdots \land  \mathsf{pl}_n a_n$, and $\mathsf{pl}_{-i} a_{-i}$ instead of $\mathsf{pl}_1 a_1 \land \cdots \land  \mathsf{pl}_{i-1} a_{i-1}\land \mathsf{pl}_{i+1} a_{i+1} \land \cdots \mathsf{pl}_n a_n$.

\subsection{Common-interpretation structures}
Let us consider first the case of common-interpretation structures. It is clear that, under common interpretation, $\gamma_i = \gamma_j$ for every $i,j \in I$. Thus we simplify things by dropping the subscript $i$. For any $a\in A$ we can write:
\begin{equation*}
\gamma(a_1, \dots, a_n) = \mu \left([[\mathsf{pl}_1 a_1 \land \cdots \land  \mathsf{pl}_n a_n]]\right).
\end{equation*}

Recall that a probability distribution $\gamma \in \Delta(A)$ is a \textbf{correlated equilibrium} of $G$ if, for every $i\in I$ and every $a_i\in A_i$,
\begin{equation*}
\sum_{a_{-i}\in A_{-i}} \left[ u_i(a_i, a_{-i}) - u_i(a_i', a_{-i})\right]\gamma (a_i, a_{-i})\geq  0\; \text{ for every } a_i'\in A_i.
\end{equation*}

We can now establish the first result.

\begin{proposition}\label{prop:From-CI-to-CE}
Let $M$ be a common-interpretation epistemic structure satisfying Assumptions \ref{ass:Sign-interpret}-\ref{ass:IR}. Then any coordination strategy induces a correlated equilibrium of $G$.
\end{proposition}
\begin{proof}
The argument is standard. Suppose $\sum_{a_{-i}\in A_{-i}} \gamma(a_i, a_{-i}) >0$. Then we have:
\begin{equation}\label{eq:Prob0}
\sum_{a_{-i}\in A_{-i}} \left[ u_i(a_i, a_{-i}) - u_i(a_i', a_{-i})\right]\gamma (a_i, a_{-i}) \propto  \sum_{a_{-i}\in A_{-i}} \left[ u_i(a_i, a_{-i}) - u_i(a_i', a_{-i})\right]\gamma (a_{-i}\vert a_i),
\end{equation}
and the right hand side of \eqref{eq:Prob0} is equal to 
\begin{align}
\sum_{a_{-i}\in A_{-i}} \left[ u_i(a_i, a_{-i}) - u_i(a_i', a_{-i})\right]\mu\left([[\mathsf{pl}_{-i} a_{-i}]]\vert [[\mathsf{pl}_i  a_{i}]]\right). \label{eq:Prob1}
\end{align}

Now we argue that the event $[[\mathsf{pl}_i  a_{i}]]$ in \eqref{eq:Prob1} is the union of some cells of $H_i$. Assumptions \ref{ass:Sign-interpret} and \ref{ass:info-partitions} imply that, for every cell $h_i \in H_i$, there exists a unique signal $\sigma \in \Sigma$ such that $(M,\omega)\vDash \mathsf{rec}_i \, \sigma$ for every $\omega \in h_i$. Combining this with Definition \ref{def:Coord} and Assumption \ref{ass:measure}, we can conclude that, for every $h_i\in H_i$, there exists a unique action $a_i' \in A_i$ such that $(M, \omega)\vDash \mathsf{pl}_i \, a_i'$ for every $\omega \in h_i$.

Since $[[\mathsf{pl}_i  a_{i}]]$ can be written as the union of some cells of $H_i$, and by the law of total probability, we can write
\begin{align*}
\mu\left([[\mathsf{pl}_{-i} a_{-i}]]\vert [[\mathsf{pl}_i  a_{i}]]\right) = \sum_{\{h_i\in H_i: h_i \subseteq [[\mathsf{pl}_{i} a_i]]\}}\mu\left([[\mathsf{pl}_{-i} a_{-i}]]\vert h_i\right)\mu\left(h_i\vert [[\mathsf{pl}_i a_i]]\right).
\end{align*}
Substituting in \eqref{eq:Prob1} and rearranging yields
\begin{equation}\label{eq:Prob2}
\sum_{\{h_i\in H_i: h_i \subseteq [[\mathsf{pl}_{i} a_i]]\}}\;\; \sum_{a_{-i}\in A_{-i}} \left[ u_i(a_i, a_{-i}) - u_i(a_i', a_{-i})\right]\times \left[\mu\left([[\mathsf{pl}_{-i} a_{-i}]]\vert h_i\right)\mu\left(h_i\vert [[\mathsf{pl}_i a_i]]\right)\right].
\end{equation}
By Assumption \ref{ass:IR}, for every $\omega \in h_i \subseteq [[\mathsf{pl}_i a_i]]$, we have that $(M,\omega)\vDash \mathsf{pl}_i a_i \land (a_i^* = a_i)$. Therefore, \eqref{eq:Prob2} is non-negative, so proving the claim.
\end{proof}

The result says that, in common-interpretation structures, self-enforcing coordination strategies always lead to an objective correlated equilibrium of the underlying game. The result can be interpreted as a syntactic version of the classical analysis of \cite{aumann1987}. The role of common-interpretation can be described as follows. Even if different agents may receive different signals in the same state, everyone agrees on the profile of actions that is being played at that state. It is never the case that $i$ thinks that $j$ is playing $a_j$ whereas $k$ thinks that $j$ is playing $b_j$ in a given state. Differently put, agents can have different information but they all share the same model or view of the world, so ruling out any form of fundamental disagreement.

The next result is about the opposite direction, namely from correlated equilibria to epistemic structures.

\begin{proposition}\label{prop:CI}
Let $\gamma$ be a correlated equilibrium of $G$. Then there exist an individually rational, common-interpretation structure $M$, a set of signals $\Sigma$, and a coordination strategy $C$ that induce $\gamma$.
\end{proposition}

\begin{proof}
Suppose $\gamma$ is a correlated equilibrium of $G$. Let $A^*\subseteq A$ be the support of $\gamma$. We define a common-interpretation structure $M = \left(\Omega, \mu, \pi, \{H_i\}_{i\in I}\right)$ by constructing one state $\omega_a$ for each action profile $a\in A^*$, so that $\Omega = \{\omega_a : a \in A^*\}$. The prior corresponds with the correlated equilibrium: for each state $\omega_a$, we set $\mu(\omega_a) = \gamma(a)$. To define the interpretation function $\pi$ and the information partitions $\{H_i\}_{i\in I}$, we first need to say more about formulas in the language $\mathcal{L}$.

Fix a set $\Sigma$ of signals such that $\vert \Sigma \vert = \max_{i \in I} \vert A_i\vert$. This allows us to choose, for each player $i\in I$, an injective function $s_i:A_i \longrightarrow \Sigma$ that we use to assign signals to players. Since each $s_i$ is injective, distinct actions correspond to different signals. The interpretation function is a function $\pi: \Omega \times \Phi \longrightarrow \{0,1\}$ such that, for all $\omega_a \in \Omega$, $i\in I$, $\sigma \in \Sigma$, and $b_i\in A_i$,

\begin{equation*}
\pi(\omega_a, \mathsf{rec}_i \sigma)=
\begin{cases}
1 & \text{ if } \sigma = s_i(a_i),\\
0 & \text{ otherwise} 
\end{cases}
\quad \text{and} \quad \pi(\omega_a, \mathsf{pl}_i b_i) = \begin{cases}
1 & \text{ if } b_i = a_i,\\
0 & \text{ otherwise.} 
\end{cases}
\end{equation*}

This implies that, for each state $\omega_a$ and each player $i$, the formula $\mathsf{rec}_i s_i(a_i) \land \mathsf{pl}_i a_i$ is deemed true at $\omega_a$. For each player $i$, the information partition $H_i$ is defined so that, for each state $\omega_a \in \Omega$, the cell $h_i(\omega_a)$ contains all the states where $i$ receives the same signal. By definition of $s_i$, we have
\begin{equation}\label{eq:inf-part-PROOF}
h_i(\omega_a) = \left\lbrace \omega_b \in \Omega : s_i(b_i)=s_i(a_i)\right\rbrace = \left\lbrace \omega_b \in \Omega : b_i = a_i\right\rbrace.
\end{equation}

One can easily verify that the structure $M$ constructed thus far satisfies Assumptions \ref{ass:Sign-interpret}-\ref{ass:measure}. To show that $M$ is individually rational, suppose that $(M,\omega_a) \vDash \mathsf{pl}_i a_i$. Player $i$'s expected payoff at $\omega_a$ from playing $a_i$ is
\begin{align*}
U_i(a_i) &= \sum_{a_{-i} \in A_{-i}} u_i(a_i, a_{-i}) \mu\left([[\mathsf{pl}_{-i} a_{-i}]] \vert h_i(\omega_a)\right)\\
&\propto \sum_{a_{-i} \in A_{-i}} u_i(a_i, a_{-i}) \mu\left([[\mathsf{pl}_{-i} a_{-i}]] \cap h_i(\omega_a)\right)\\
&=\sum_{a_{-i} \in A_{-i}} u_i(a_i, a_{-i}) \mu\left([[\mathsf{pl}_{-i} a_{-i}]] \cap [[\mathsf{pl}_i a_i]]\right)\\
&=\sum_{a_{-i} \in A_{-i}} u_i(a_i, a_{-i}) \mu\left([[\mathsf{pl}_1 a_1 \land \dots \land \mathsf{pl}_n a_n]]\right)\\
&= \sum_{a_{-i} \in A_{-i}} u_i(a_i, a_{-i}) \gamma\left(a_i, a_{-i} \right).
\end{align*}

Therefore, since $\gamma$ is a correlated equilibrium by assumption, we can conclude that $(M,\omega) \vDash a_i^* = a_i$.

Finally, we need to construct a coordination strategy $C$ that induces $\gamma$. For each $i\in I$ and each $\sigma \in \Sigma$, if $\sigma$ is in the range of $s_i$, then the formula $\mathsf{rec}_i \sigma \implies \mathsf{pl}_i a_i$, with $a_i = s_i^{-1}(\sigma)$, is in $C$. If $\sigma$ is not in the range of $s_i$, then pick an arbitrary $a_i'\in A_i$ and add the formula $\mathsf{rec}_i \sigma \implies \mathsf{pl}_i a'_i$ to $C$. One can easily verify that $C$ is indeed a coordination strategy as per Definition \ref{def:Coord}, and that it induces the correlated equilibrium $\gamma$.
\end{proof}

In the following example, we illustrate the construction that we have just used in proving Proposition \ref{prop:CI}.

\paragraph{Example 2} Consider the base game $G$ in Figure \ref{fig1}.

\begin{figure}[ht]
\centering
\begin{game}{3}{3}
             & $L$      & $C$  & $R$ \\
$T$       &$0,0$    &$2,1$ & $1,2$ \\
$M$      &$1,2$    &$0,0$ &$2,1$\\
$B$       &$2,1$    &$1,2$ &$0,0$
\end{game}
\caption{The base game. \label{fig1}}
\end{figure}

This game has a unique Nash equilibrium in which either player randomizes uniformly over her available strategies. Consider the correlated equilibrium $\gamma$ that puts weight $\frac{1}{6}$ on every action profile which gives strictly positive payoffs. We want to find an individually rational, common-interpretation structure that induces such an equilibrium. We start by noticing that the support of $\gamma$ is the following:
\begin{equation*}
A^* = \left\lbrace (T,C), (T,R), (M, L), (M,R), (B,L), (B, C) \right\rbrace.
\end{equation*}

Let the state space be $\Omega = \{\omega_a : a\in A^*\}$. The common prior over $\Omega$ is uniform. Fix a set of signals $\Sigma = \{\sigma, \sigma',\sigma''\}$. We assign signals to players through functions $s_i:A_i \longrightarrow \Sigma$, with $i=1,2$, such that
\begin{equation*}
(s_1(T), s_1(M), s_1(B)) = (s_2(L), s_2(C), s_2(R)) = (\sigma, \sigma', \sigma'').
\end{equation*}

The interpretation function is a function $\pi:\Omega \times \Phi \longrightarrow \{0,1\}$ which satisfies the truth assignments contained in the following table.

\begin{center}
\begin{tabular}{|l|| *{6}{c|}}
\hline
$\boldsymbol{\pi}$& $\omega_{TC}$ & $\omega_{TR}$ & $\omega_{ML}$ & $\omega_{MR}$ & $\omega_{BL}$ & $\omega_{BC}$\\ \hline \hline
$\mathsf{rec}_1 \sigma$ & 1 & 1 & 0 & 0 & 0 & 0\\ \hline
$\mathsf{rec}_1 \sigma'$ & 0 & 0 & 1 & 1 & 0 & 0\\ \hline
$\mathsf{rec}_1 \sigma''$ & 0 & 0 & 0 & 0 & 1 & 1\\ \hline
$\mathsf{rec}_2 \sigma$ & 0 & 0 & 1 & 0 & 1 & 0\\ \hline
$\mathsf{rec}_2 \sigma'$ & 1 & 0 & 0 & 0 & 0 & 1\\ \hline
$\mathsf{rec}_2 \sigma''$ & 0 & 1 & 0 & 1 & 0 & 0\\ \hline
$\mathsf{pl}_1 T$ & 1 & 1 & 0 & 0 & 0 & 0\\ \hline
$\mathsf{pl}_1 M$ & 0 & 0 & 1 & 1 & 0 & 0\\ \hline
$\mathsf{pl}_1 B$ & 0 & 0 & 0 & 0 & 1 & 1\\ \hline
$\mathsf{pl}_2 L$ & 0 & 0 & 1 & 0 & 1 & 0\\ \hline
$\mathsf{pl}_2 C$ & 1 & 0 & 0 & 0 & 0 & 1\\ \hline
$\mathsf{pl}_2 R$ & 0 & 1 & 0 & 1 & 0 & 0\\ \hline
\end{tabular}
\end{center}

By \eqref{eq:inf-part-PROOF}, information partitions are defined as follows:
\begin{align*}
H_1 &= \left\lbrace \{\omega_{TC}, \omega_{TR}\}, \{\omega_{ML}, \omega_{MR}\}, \{\omega_{BL}, \omega_{BC}\}\right\rbrace\\
H_2 &= \left\lbrace \{\omega_{TC}, \omega_{BC}\}, \{\omega_{TR}, \omega_{MR}\}, \{\omega_{ML}, \omega_{BL}\}\right\rbrace.
\end{align*}

It is straightforward to verify that the structure $M$ constructed so far is individually rational.

Finally, the coordination strategy $C$ that induces $\gamma$ can be defined as follows:
\begin{equation*}
C=\left\lbrace \mathsf{rec}_i s_i(a_i) \implies \mathsf{pl}_i a_i : i \in \{1,2\} \text{ and } a_i\in A_i \right\rbrace.
\end{equation*}

Notice that $C$ contains six formulas, and all of them are true in \textit{every} state of the world. This means that it is never the case that, say, player $1$ receives signal $\sigma'$ and plays action $T$.

\subsection{Ambiguous structures}
We now characterize the equilibrium distributions induced by epistemic structures that are possibly ambiguous. Recall that a profile of probability distributions $\left(\gamma_1, \dots, \gamma_n \right)$ over $A$ is a \textbf{subjective correlated equilibrium} of $G$ if, for every $i\in I$ and every $a_i\in A_i$,
\begin{equation*}
\sum_{a_{-i}\in A_{-i}} \left[ u_i(a_i, a_{-i}) - u_i(a_i', a_{-i})\right]\gamma_i (a_i, a_{-i})\geq  0\; \text{ for every } a_i'\in A_i.
\end{equation*}

Then we have the following.

\begin{proposition}
Let $M$ be an epistemic structure satisfying Assumptions \ref{ass:Sign-interpret}-\ref{ass:IR}. Then any coordination strategy induces a subjective correlated equilibrium of $G$.
\end{proposition}
\begin{proof}
The argument is the same as in the proof of Proposition \ref{prop:From-CI-to-CE} with the proviso that, for every $i\in I$, one uses the following decomposition of conditional probabilities:
\begin{align*}
\gamma_i (a_{-i}\vert a_i) &= \mu \left([[\mathsf{pl}_{-i} a_{-i}]]_i \vert [[\mathsf{pl}_i a_i]]_i\right)\\[15pt]
&= \sum_{\{h_i\in H_i: h_i \subseteq [[\mathsf{pl}_i a_i]]_i\}}\mu\left([[\mathsf{pl}_{-i} a_{-i}]]_i\vert h_i\right)\mu\left(h_i\vert [[\mathsf{pl}_i a_i]]_i\right).
\end{align*}
\end{proof}

The result can be interpreted as follows. Even if players agree on a coordination strategy, and even if they share a common prior, language ambiguity may cause them to ascribe different probabilities to the same event, so leading to inconsistent beliefs. Players may disagree on what action profile is being played in a given state. 
Contrary to common-interpretation structures, it may well be the case that $i$ thinks that $j$ is playing $a_j$ whereas $k$ thinks that $j$ is playing $b_j$ in a given state. Differently put, agents may have different views of the world stemming from a fundamental disagreement about the interpretation of (some) primitive propositions. We illustrate this point in the next example, where we describe an ambiguous structure whose induced equilibrium distributions are a subjective correlated equilibrium but not an objective correlated equilibrium. 

\paragraph{Example 3} Consider the elementary coordination game in Figure \ref{figX}.

\begin{figure}[ht]
\centering
\begin{game}{2}{2}
         & $L$    & $R$ \\
$U$      & $1,1$  & $0,0$ \\
$D$      & $0,0$  & $1,1$
\end{game}
\caption{The base game. \label{figX}}
\end{figure}

Suppose that the epistemic structure is the same as in Example 4.3 of \cite{halpern2015}. The state space is $\Omega = \{\omega, \omega'\}$, and the common prior is uniform. The set of signals is $\Sigma = \{\sigma, \sigma'\}$. Players disagree on the interpretation of signals. Values of their interpretation functions are reported in the following tables.

\begin{center}
\begin{tabular}{|l|| *{2}{c|}}
\hline
$\boldsymbol{\pi_1}$& $\omega$ & $\omega'$\\ \hline \hline
$\mathsf{rec}_1 \sigma$ & 1 & 0\\ \hline
$\mathsf{rec}_1 \sigma'$ & 0 & 1\\ \hline
$\mathsf{rec}_2 \sigma$ & 1 & 0\\ \hline
$\mathsf{rec}_2 \sigma'$ & 0 & 1\\ \hline
$\mathsf{pl}_1 U$ & 1 & 0\\ \hline
$\mathsf{pl}_1 D$ & 0 & 1\\ \hline
$\mathsf{pl}_2 L$ & 1 & 0\\ \hline
$\mathsf{pl}_2 R$ & 0 & 1\\ \hline
\end{tabular}
\qquad \qquad
\begin{tabular}{|l|| *{2}{c|}}
\hline
$\boldsymbol{\pi_2}$& $\omega$ & $\omega'$\\ \hline \hline
$\mathsf{rec}_1 \sigma$ & 1 & 1\\ \hline
$\mathsf{rec}_1 \sigma'$ & 0 & 0\\ \hline
$\mathsf{rec}_2 \sigma$ & 1 & 1\\ \hline
$\mathsf{rec}_2 \sigma'$ & 0 & 0\\ \hline
$\mathsf{pl}_1 U$ & 1 & 1\\ \hline
$\mathsf{pl}_1 D$ & 0 & 0\\ \hline
$\mathsf{pl}_2 L$ & 1 & 1\\ \hline
$\mathsf{pl}_2 R$ & 0 & 0\\ \hline
\end{tabular}
\end{center}

Thus we have:
\begin{align*}
[[\mathsf{rec}_1 \sigma]]_1 = \{\omega\} &&  [[\mathsf{rec}_2 \sigma]]_1 = \{\omega\}\\
[[\mathsf{rec}_1 \sigma']]_1 = \{\omega'\} && [[\mathsf{rec}_2 \sigma']]_1 = \{\omega'\}
\end{align*}
for player $1$ and 
\begin{align*}
[[\mathsf{rec}_1 \sigma]]_2 = \{\omega, \omega'\} &&  [[\mathsf{rec}_2 \sigma]]_2 = \{\omega, \omega'\}\\
[[\mathsf{rec}_1 \sigma']]_2 = \emptyset && [[\mathsf{rec}_2 \sigma']]_2 = \emptyset
\end{align*}
for player $2$. In words, each agent thinks that the other always receives the same signal as hers. Information partitions are given by:
\begin{align*}
H_1 &= \left\lbrace \{\omega\}, \{\omega'\}\right\rbrace\\
H_2 &= \left\lbrace \{\omega, \omega'\}\right\rbrace.
\end{align*}

Now suppose that the coordination strategy $C$ contains the following four formulas:
\begin{align*}
\mathsf{rec}_1 \sigma &\implies \mathsf{pl}_1\; U\\
\mathsf{rec}_2 \sigma &\implies \mathsf{pl}_2\; L\\
\mathsf{rec}_1 \sigma' &\implies \mathsf{pl}_1\; D\\
\mathsf{rec}_2 \sigma' &\implies \mathsf{pl}_2 \; R.
\end{align*}

One can verify that $C$ is self-enforcing. The induced subjective probability distributions over $A$ are $\gamma_1(U,L)= \gamma_1(D,R)= \frac{1}{2}$ and $\gamma_2(U,L)= 1$. Finally, we observe that $(M, \omega',2)\vDash \mathsf{rec}_1 \sigma \land \mathsf{pl}_1 U \land \lnot (a_1^* = U)$. In words, $2$ thinks that $1$ is choosing an action that is not utility-maximizing. The reason is that, as we pointed out in Remark \ref{remarkCB}, individual rationality in ambiguous structures does not entail common belief in rationality.

The next result is about the opposite direction, namely from subjective correlated equilibria to epistemic structures that induce them.

\begin{proposition}
Let $\left(\gamma_1, \dots, \gamma_n\right)$ be a subjective correlated equilibrium of $G$. Then there exist an individually rational epistemic structure $M$, a set of signals $\Sigma$, and a coordination strategy $C$ that induce $\left(\gamma_1, \dots, \gamma_n\right)$.
\end{proposition}
\begin{proof}
The argument follows the same logic as in the proof of Proposition \ref{prop:CI}. Suppose $\left(\gamma_1, \dots, \gamma_n\right)$ is a subjective correlated equilibrium of $G$. For every $i\in I$, let $A_*^i \subseteq A$ be the support of $\gamma_i$. Define $A_* := \times_{i\in I} A_*^i$. Elements $a$ in $A_*$ are profiles of action profiles, and we use the following notation:
\begin{equation*}
a = \left(a^1, \dots, a^n\right) =\left((a_1^1, \dots, a_n^1), \dots, (a_1^n,\dots, a_n^n)\right).
\end{equation*}

We define an epistemic probability structure $M=\left(\Omega, \mu, \{\pi_i\}_{i\in I}, \{H_i\}_{i\in I}\right)$ as follows. We construct one state $\omega_a$ for each element $a\in A_*$, so that $\Omega = \{\omega_a : a\in A_*\}$. The common prior $\mu$ over $\Omega$ is constructed as a product measure: for every $\omega_a$, we let $\mu(\omega_a) = \prod_{i = 1}^{n}\gamma_i(a^i)$. 

Now fix a set $\Sigma$ of signals such that $\vert \Sigma \vert = \max_{i \in I} \vert A_i\vert$. For each $i\in I$, we can choose an injective function $s_i: A_i \longrightarrow \Sigma$ that we use to assign signals to players. Since each $s_i$ is injective, distinct actions correspond to different signals. For each $i\in I$, the interpretation function is a function $\pi_i: \Omega \times \Phi \longrightarrow \{0,1\}$ such that, for all $\omega_a\in \Omega$, $j\in I$, $\sigma \in \Sigma$, and $b_j\in A_j$,
\begin{equation}\label{eq:interpret-funct}
\pi_i(\omega_a, \mathsf{rec}_j \sigma)=
\begin{cases}
1 & \text{ if } \sigma = s_j(a_j^i),\\
0 & \text{ otherwise} 
\end{cases}
\quad \text{and} \quad \pi_i(\omega_a, \mathsf{pl}_j b_j) = \begin{cases}
1 & \text{ if } b_j = a_j^i,\\
0 & \text{ otherwise.} 
\end{cases}
\end{equation}

The above definition implies that, for each state $\omega_a$, and each player $i$, the formulas $\left\lbrace \mathsf{rec}_j s_j(a_j^i) \land \mathsf{pl}_j a_j^i : j\in I\right\rbrace$ are deemed as true at $\omega_a$ by player $i$. For each $i\in I$, the information partition $H_i$ is defined so that, for each state $\omega_a\in \Omega$, the cell $h_i(\omega_a)$ contains all the states where, according to $i$, $i$ receives the same signal. By definition of $s_i$, we have
\begin{equation}\label{eq:inf-part-PROOF2}
h_i(\omega_a) = \left\lbrace \omega_b \in \Omega : s_i(b_i^i)=s_i(a_i^i)\right\rbrace = \left\lbrace \omega_b \in \Omega : b_i^i = a_i^i\right\rbrace.
\end{equation}

One can easily verify that the structure $M$ constructed thus far satisfies Assumptions \ref{ass:Sign-interpret}-\ref{ass:measure}. To show that $M$ is individually rational, suppose that $(M,\omega_a,i) \vDash \mathsf{pl}_i a_i$. By \eqref{eq:interpret-funct}, $a_i=a_i^i$. Player $i$'s expected payoff at $\omega_a$ from playing $a_i$ is
\begin{align}
U_i(a_i) &= \sum_{a_{-i} \in A_{-i}} u_i(a_i, a_{-i}) \mu\left([[\mathsf{pl}_{-i} a_{-i}]]_i \vert h_i(\omega_a)\right) \nonumber \\
&\propto \sum_{a_{-i} \in A_{-i}} u_i(a_i, a_{-i}) \mu\left([[\mathsf{pl}_{-i} a_{-i}]]_i \cap h_i(\omega_a)\right)\nonumber\\
&=\sum_{a_{-i} \in A_{-i}} u_i(a_i, a_{-i}) \mu\left([[\mathsf{pl}_{-i} a_{-i}]]_i \cap [[\mathsf{pl}_i a_i]]_i\right) \nonumber\\
&=\sum_{a_{-i} \in A_{-i}} u_i(a_i, a_{-i}) \mu\left([[\mathsf{pl}_1 a_{1} \land \dots \land \mathsf{pl}_n a_n]]_i\right) \nonumber\\
&=\sum_{a_{-i} \in A_{-i}} u_i(a_i, a_{-i}) \mu\left(\{\omega_b\in \Omega: b^{i} = (a_i, a_{-i})\}\right).\label{eq:ProbProofAmb}
\end{align}

By definition of the common prior $\mu$, we have that $\mu\left(\{\omega_b\in \Omega: b^{i} = (a_i, a_{-i})\}\right)$ is equal to
\begin{equation}\label{eq:Prior-decomp}
\gamma_i(a_i, a_{-i})\left[\sum_{(b^1,\dots, b^{i-1}, b^{i+1}, \dots, b^n )\in \times_{j\neq i} A_*^j} \gamma_1(b^1)\times\cdots \times \gamma_{i-1}(b^{i-1})\times\gamma_{i+1}(b^{i+1})\times \cdots \times \gamma_n(b^n)\right],
\end{equation}
which simplifies to $\gamma_i(a_i, a_{-i})$. Substituting in \eqref{eq:ProbProofAmb}, we obtain
\begin{equation*}
U_i(a_i) \propto \sum_{a_{-i} \in A_{-i}} u_i(a_i, a_{-i}) \gamma_i\left(a_i, a_{-i} \right).
\end{equation*}
Therefore, since $\gamma_i$ is part of a subjective correlated equilibrium by assumption, we can conclude that $(M,\omega_a,i) \vDash \mathsf{pl}_i a_i \land (a_i^* = a_i)$.

It remains to construct a coordination strategy $C$ that induces $\left(\gamma_1, \dots, \gamma_n\right)$. For each $i\in I$ and each $\sigma \in \Sigma$, if $\sigma$ is in the range of the injective function $s_i$, then the formula $\mathsf{rec}_i \sigma \implies \mathsf{pl}_i a_i$, with $a_i = s_i^{-1}(\sigma)$, is in $C$. If $\sigma$ is not in the range of $s_i$, then pick an arbitrary $a_i'\in A_i$ and add the formula $\mathsf{rec}_i \sigma \implies \mathsf{pl}_i a_i'$ to $C$. One can easily verify that $C$ is indeed a coordination strategy as per Definition \ref{def:Coord}. Finally, to show that it induces the subjective correlated equilibrium $\left(\gamma_1, \dots, \gamma_n\right)$, by using \eqref{eq:Prior-decomp}, we have that, for every $i\in I$ and every $(a_1,\dots, a_n)\in A$,
\begin{equation*}
\mu\left([[\mathsf{pl}_{1}a_1\land \dots \land  \mathsf{pl}_n a_{n}]]_i\right) =  \mu\left(\{\omega_b\in \Omega: b^{i} = (a_1, \dots, a_n)\}\right) = \gamma_i(a_1, \dots, a_n),
\end{equation*}
so ending the proof.
\end{proof}

\section{Discussion}
We use the word ambiguity as it is done in linguistics, where it expresses the fact that the map from sentences to meanings is multi-valued. Our analysis has nothing to do with ambiguity in the decision-theoretic sense of not knowing the ``true'' probability distribution of a certain event. We take ambiguity as a given and do not model why different players can assign different truth values to the very same formula. Our interpretation is that ambiguity is a structural property of natural language. Differently put, the map from sentences to meanings induced by any natural language is not commonly known. The gist of our analysis is that players can agree on sentences and, at the same time, disagree on meanings. When strategic interaction is conditioned on sentences, it is the uncertainty about their meanings that acts as a correlating device.

Introduced by \cite{aumann1974}, the standard model for correlated equilibrium is set in an event-based epistemic framework. That is, players reason about events which are represented as subsets of a given state space. The language through which agents describe events is not modeled explicitly. Adopting a syntactic approach, our analysis consists in enriching the standard model for correlated equilibrium with a formal language. As a consequence, agents' reasoning about the world and, in particular, the game they are going to play is now expressed through formulas; the state space is a representation of how agents assign meaning to formulas. The standard model can be seen as a reduced form model of the syntactic approach we use. A comparison between the event-based and the syntactic epistemic frameworks, but without game theoretic applications and without language ambiguity, can be found in \cite{halpern2003reasoning}.

Players are assumed to be fully rational. Even if they have different interpretation functions, they fully understand that the interpretation of probability formulas, hence beliefs, is not the same for everyone. In addition, their information is always partitional. We refer to \cite{brandenburger1992GenCorrelated} for a construction of correlated equilibrium with boundedly rational players. In their model, players make systematic mistakes in processing information, so leading to non-partitional information functions. They show that information processing errors are equivalent to introduce ``subjectivity'' in beliefs. What we show in this paper is that, without any information processing error, ambiguity in natural language provides a justification for heterogeneous beliefs.

At first blush, it might be surprising that players having a common prior over a fixed state space end up having different subjective beliefs about their play in the underlying game. The reason why this is the case can be explained as follows. The state space $\Omega$ can be seen as a collection of $n$ models about the world. When all of these models are exactly the same, then their ``projection'' over $A$ is obviously the same for everybody. But when the subjective models differ, different players might have different projections over $A$, because the event $(a_1,\dots, a_n)$ is not the same for everybody (i.e. the set of states where it holds is not the same for everyone). Differently put, language ambiguity induces subjectivity in how players reason about their choices in $G$.

\section{Conclusion}
We have examined how players can coordinate their choices when the language through which they communicate is possibly ambiguous. The gist of our results is that, when players publicly agree to condition their play on a set of sentences, the meaning of these sentences is not necessarily commonly known because of language ambiguity. The resulting uncertainty acts as a correlating device, so inducing correlated equilibrium distributions over outcomes. We believe that our analysis also provides a justification for why subjective beliefs about strategic play may not be consistent.

\bibliographystyle{plainnat}
\bibliography{biblio}

\end{document}